\documentclass[a4paper,UKenglish]{lipics-v2018}

\usepackage{todonotes}

\usepackage{xspace,xargs}
\usepackage[utf8]{inputenc}

\title{A New Class of Searchable and Provably Highly Compressible String Transformations}

\titlerunning{A New Class of Searchable and Highly Compressible String Transformations}

\author{Raffaele Giancarlo}{University of Palermo, Dipartimento di Matematica e Informatica, Italy}{}{}{}

\author{Giovanni Manzini}{University of Eastern Piedmont, Alessandria, and IIT-CNR, Pisa, Italy}{}{}{}

\author{Giovanna Rosone}{University of Pisa, Dipartimento di Informatica, Italy}{}{}{}

\author{Marinella Sciortino}{University of Palermo, Dipartimento di Matematica e Informatica, Italy}{}{}{}

\authorrunning{R. Giancarlo and G. Manzini and G. Rosone and M. Sciortino}

\Copyright{Raffaele Giancarlo and Giovanni Manzini and Giovanna Rosone and Marinella Sciortino}

\long\def\ignore#1{\vskip 0pt}

\subjclass{
; 
}

\category{}

\ignore{---------------------------------------------

\relatedversion{}

\supplement{}

\funding{Partially supported by: MIUR-SIR project CMACBioSeq \vir{Combinatorial methods for analysis and compression of biological sequences} grant n.~RBSI146R5L; INdAM - GNCS Project 2018 \vir{Analysis and Processing of Big Data based on Graph Models}; MIUR-PRIN grant 201534HNXC.}

--------------------------------------------------}

\EventEditors{John Q. Open and Joan R. Access}
\EventNoEds{2}
\EventLongTitle{42nd Conference on Very Important Topics (CVIT 2016)}
\EventShortTitle{CVIT 2016}
\EventAcronym{CVIT}
\EventYear{2016}
\EventDate{December 24--27, 2016}
\EventLocation{Little Whinging, United Kingdom}
\EventLogo{}
\SeriesVolume{42}
\ArticleNo{23}
\nolinenumbers 
\hideLIPIcs  

\newcommand{\vir}[1]{``#1''}

\newcommand{\BWT}{BWT\xspace}
\newcommand{\BWTs}{BWTs\xspace}

\newcommand{\ABWT}{ABWT\xspace}

\newcommand{\rank}{\ensuremath{\mathsf{rank}}\xspace}
\newcommand{\sel}{\ensuremath{\mathsf{select}}\xspace}
\newcommand{\access}{\ensuremath{\mathsf{access}}\xspace}

\newcommand{\Rng}[1]{R[#1]}
\newcommand{\Rngx}[1]{R^{\ast}[#1]}
\newcommand{\Mx}{M_{*}}
\newcommand{\BWTx}{BWT_{*}}
\newcommand{\asize}{{\sigma}}
\newcommand{\A}{\Sigma}
\newcommand{\Oh}{{\cal O}}

\newcommand{\xchar}{{c}}
\newcommand{\bone}{{\bf 1}}
\newcommand{\btwo}{{\bf 2}}
\newcommand{\bthree}{{\bf 3}}
\long\def\ignore#1{}

\begin{document}

\maketitle

\begin{abstract}
The Burrows-Wheeler Transform is a string transformation that plays a fundamental role for the design of self-indexing compressed data structures.  Over the years, researchers have successfully extended this transformation outside the domains of strings. However, efforts to find non-trivial alternatives of the original, now 25 years old, Burrows-Wheeler string transformation have met limited success.
In this paper we bring new lymph to this area by introducing a whole new family of transformations that have all the \vir{myriad virtues} of the BWT: they can be computed and inverted in linear time, they produce provably highly compressible strings, and they support linear time pattern search directly on the transformed string. This new  family is a special case of a more general class of transformations based on {\em context adaptive alphabet orderings}, a concept introduced here. This more general class includes also the Alternating BWT, 
another invertible string transforms recently introduced in connection with a generalization of Lyndon words. 
\end{abstract}

\section{Introduction}

The Burrows Wheeler Transform~\cite{bwt94} (\BWT) is a string transformation that had a revolutionary impact in the design of succinct or compressed data structures.  Originally proposed as a tool for text compression, shortly after its introduction~\cite{Ferragina:2000} it has been shown that, in addition to making easier to represent a string in space close to its entropy, it also makes easier to search for pattern occurrences in the original string. After this discovery, data transformations inspired by the \BWT have been proposed for compactly represent and search other combinatorial objects such as: trees, graphs, finite automata, and even string alignments. See~\cite{tcs/GagieMS17} for an attempt to unify some of these results and~\cite{Nav16} for an in-depth treatment of the field of compact data structures.

Going back to the original Burrows-Wheeler string transformation, we can summarize its salient features as follows: \bone) it can be  computed and inverted in linear time, \btwo) it produces strings which are provably compressible in terms of the high order entropy of the input, \bthree) it supports pattern search directly on the transformed string in time proportional to the pattern length. It is the {\em combination} of these three properties that makes the \BWT a fundamental tool for the design of compressed self-indices. In Section~\ref{sec:prel} we review these properties and also the many attempts to modify the original design. However, we recall that, despite more than twenty years of intense scrutiny, the only non trivial known \BWT variant that fully satisfies properties \bone--\bthree\ is the {\em Alternating \BWT} (\ABWT). The \ABWT has been introduced in~\cite{GesselRestivoReutenauer2012} in the field of combinatorics of words and its basic algorithmic properties have been described in~\cite{dlt/GiancarloMRRS18}.


In this paper we introduce a new {\em whole family} of transformations that satisfy properties \bone--\bthree\ and can therefore replace the \BWT in the construction of compressed self-indices with the same time efficiency of the original \BWT and the potential of achieving better compression. We show that our family, supporting linear time computation, inversion, and search, is a special case of a much larger class of transformations that also satisfy properties \bone--\bthree\ except that, in the general case, inversion and pattern search may take quadratic time. Our larger class includes as special cases also the \BWT and the \ABWT and therefore it constitutes a natural candidate for the study of additional properties shared by all known \BWT variants.

\ignore{This class was mentioned in~\cite{FGMS2005}, where it was observed that these transformations can be computed in linear time and that they produce highly compressible strings (in the sense of property \btwo).  
However, in that paper, the problem of the invertibility of these transformations was left open: here we prove that they are all invertible by providing a quadratic time inversion algorithm. In addition, we show that these transformations support pattern search directly in the transformed text in time quadratic with the pattern length.}

More in detail, in Section~\ref{sec:genBWT} we describe a class of string transformations based on {\em context adaptive alphabet orderings}. 
The main feature of the above class of transformations is that, in the rotation sorting phase, we use alphabet orderings that depend on the context (i.e., the longest common prefix of the rotations being compared). In Section~\ref{sec:local} we consider the subclass of transformations based on {\em local orderings}. In this subclass, the alphabet orderings only depend on a constant portion of the context. We prove that local ordering transformations can be inverted in linear time, and that pattern search in the transformed string takes time proportional to the pattern length. Thus, these transformations have the same properties \bone--\bthree\ that were so far prerogative of the \BWT and \ABWT.

\ignore{The choice can now depend on the string (or family of strings) being processed and thus opens a new challenging scenario in \BWT-based compression. Given the large number of data structures featuring the \BWT we believe that this area of research can have significant impact in theory and practice.}

Having now at our disposal a wide class of string transformations with the same remarkable properties of the \BWT, it is natural to use them to improve \BWT-based data structures by selecting the one more suitable for the task.  In this paper we initiate this study by considering the problem of selecting the \BWT variant that minimizes the number of runs in the transformed string. The motivation is that data centers often store highly repetitive collections, such as genome databases, source code repositories, and versioned text collections. For such highly repetitive collections there is theoretical and practical evidence that the entropy underestimates the compressibility of the collection and much better compression ratios are obtained exploiting runs of equal symbols in the \BWT~\cite{bioinformatics/CoxBJR12,soda/GagieNP18,stoc/KempaP18,tcs/KreftN13,jcb/MakinenNSV10,MantaciRRS17,MantaciRRSV17}.
In Section~\ref{sec:runs} we show that, for constant size alphabet,for the most general class of transformations considered in this paper, the \BWT variant that minimizes the number of runs can be found in linear time using a dynamic programming algorithm. 

\ignore{journal

Although such result does not lead to a practical compression algorithm, such minimal number of runs constitutes a lower bound for the number of runs achievable by the other variants described in this paper and therefore constitutes a baseline for further theoretical or experimental studies.}

\section{Notation and background}\label{sec:prel}

Let $\A =\{c_1, c_2, \ldots, c_\asize\}$ be a finite ordered alphabet of size $\asize$ with $c_1< c_2< \cdots < c_\asize$, where $<$ denotes the standard lexicographic order. We denote by $\A^*$ the set of strings over $\A$.  Given a string $x=x_1 x_2 \cdots x_n \in \A^*$ we denote by $|x|$ its length $n$. We use $\epsilon$ to denote the empty string.

A \emph{factor} of $x$ is written as $w[i,j] = w_i \cdots w_j$ with $1\leq i \leq j \leq n$.  A factor of type $w[1,j]$ is called a \emph{prefix}, while a factor of type $w[i,n]$ is called a \emph{suffix}. The $i$-th symbol in $x$ is denoted by $x[i]$.  Two strings $x,y\in \Sigma^*$ are called {\em conjugate}, if $x=uv$ and $y=vu$, where $u,v\in \Sigma^*$. We also say that $x$ is a {\em cyclic rotation} of $y$. A string $x$ is {\em primitive} if all its cyclic rotations are distinct.
Given a string $x$ and $c\in\A$, we write $\rank_c(x,i)$ to denote the number of occurrences of $c$ in $x[1,i]$, and $\sel_c(x,j)$ to denote the position of the $j$-th $c$ in~$x$. 

Given a primitive string $s$, we consider the matrix of all its cyclic rotations sorted in lexicographic order. Note that the rotations are all distinct by the primitivity of~$s$. The last column of the matrix is called the Burrows-Wheeler Transform of the string $s$ and it is denoted by $\BWT(s)$ (see Fig.~\ref{fig:ABWTx} left). The \BWT can be computed in $\Oh(|s|)$ time using any algorithm for Suffix Array construction~\cite{GIA07,KSB06}. It is shown in~\cite{bwt94} that $\BWT(s)$ is always a permutation of $s$, and that there exists a linear time procedure to recover $s$ given $\BWT(s)$ and the position $I$ of $s$ in the rotations matrix (it is $I=2$ in Fig.~\ref{fig:BWTx} left).


The \BWT has been introduced as a data compression tool: it was empirically observed that $\BWT(s)$ usually contains long runs of equal symbols. This notion was later mathematically formalized in terms of the empirical entropy of the input string~\cite{FGMS2005,Manzini2001}. For $k \geq 0$, the $k$-th order empirical entropy of a string $x$, denoted as $H_k(x)$, is a lower bound to the compression ratio of any algorithm that encodes each symbol of $x$ using a codeword that only depends on the $k$ symbols preceding it in $x$. The simplest compressors, such as Huffman coding, in which the code of a symbol does not depend on the previous symbols, typically achieve a (modest) compression bounded in terms of the zeroth-oder entropy $H_0$. This class of compressors are referred to as {\em memoryless} compressors. 

\ignore{Journal 
More sophisticated compressors, such as Lempel-Ziv compressors and derivatives, use knowledge from the already seen part of the input to compress the incoming symbols. They are slower than memoryless compressors but they achieve a much better compression ratio, which can be usually bounded in terms of the $k$-th order entropy of the input string for a large $k$~\cite{koma00}. }

It is proven in~\cite[Theorem 5.4]{FGMS2005} that the informal statement ``the output of the \BWT is {highly compressible}'' can be formally restated saying that $\BWT(s)$ can be compressed up to $H_k(s)$, for any $k>0$, using any tool able to compress up to the zeroth-order entropy. In other words, after applying the \BWT we can achieve high order compression using a simple (and fast) memoryless compressor. This property is often referred to as the ``boosting'' property of the \BWT.
Another remarkable property of the \BWT is that it can be used to build compressed indices. It is shown in~\cite{Ferragina:2005} how to compute the number of occurrences of a pattern $x$ in $s$ in $\Oh(t_R|x|)$ time, where $t_R$ is the cost of executing a $\rank$ query over $\BWT(s)$. This result has spurred a great interest in data structures representing compactly a string $x$ and efficiently supporting the queries \rank, \sel, and \access (return $x[i]$ given $i$, which is a nontrivial operation when $x$ is represented in compressed form) and there are now many alternative solutions with different trade-offs. In this paper we assume a RAM model with word size $w$ and an alphabet of size $\sigma = w^{\Oh(1)}$. Under this assumption we make use of the following result (Theorem~7 in~\cite{BNtalg14})

\begin{theorem}\label{theo:rsa}
Let $s$ denote a string over an alphabet of size $\sigma= w^{\Oh(1)}$. We can represent~$s$ in $|s| H_0(s) + o(|s|)$ bits and support constant time \rank, \sel, and \access queries.\qed
\end{theorem}

The properties of the \BWT of being {\em compressible} and {\em searchable} combine nicely to give us {\it indexing capabilities} in {\it compressed space}. Indeed, combining a zero order representation supporting \rank, \sel, and \access queries with the boosting property of the \BWT, we obtain a full text self-index for $s$ that uses space bounded by $|s|H_k(s) + o(|s|)$ bits; see~\cite{Ferragina:2005,Mak15,Nav16,NM-survey07} for further details on these results and on the field of compressed data structures and algorithms that originated from this area of research.

\ignore{We refer to the one in~\cite{BNtalg14}, where, assuming a RAM model with word size $w$, for an alphabet of size $\sigma = w^{\Oh(1)}$ it is shown how to represent any string $x$ in $|x| H_0(x) + o(|x|)$ bits and support constant time \rank, \sel, and \access queries.}

\ignore{
\begin{theorem}
Given a memoryless compressor $A$, compressing each string $x$ in up to $\lambda|x|H_0(x) + \mu |x|$, we can find in $\Oh(s)$ time a partition $x_1 x_2 \cdots x_h$ of $\BWT(s)$  
\end{theorem}}

\subsection{Known BWT variants}

We observed that the salient features of the Burrows-Wheeler transformation can be summarized as follows: \bone) it can be  computed and inverted in linear time, \btwo) it produces strings which are provably compressible in terms of the high order entropy of the input, \bthree) it supports linear time pattern search directly on the transformed string. The {\em combination} of these three properties makes the \BWT a fundamental tool for the design of compressed self-indices.  Over the years, many variants of the original \BWT have been proposed; in the following we review them, in roughly chronological order, emphasizing to what extent they share the features \bone--\bthree\ mentioned above.  

The original \BWT is defined by sorting in lexicographic order all the cyclic rotations of the input string. In~\cite{Schindler1997} Schindler proposes a {\em bounded context} transformation that differs from the \BWT in the fact that the rotations are lexicographically sorted considering only the first $\ell$ symbols of each rotation. Recent studies~\cite{spe/CulpepperPP12,ccp/PetriNCP11} have shown that this variant satisfies properties \bone--\bthree, with the limitation that the compression ratio can reach at maximum the $\ell$-th order entropy and that it supports searches of patterns of length at most $\ell$. Chapin and Tate~\cite{dcc/ChapinT98} have experimented with computing the \BWT using a different alphabet order. This simple variant still satisfies properties \bone--\bthree, but it clearly does not bring any new theoretical insight. More recently, some authors have proposed variants in which the lexicographic order is replaced by a different order relation. The interested reader can find relevant work in a recent review~\cite{DAYKIN2017}; it turns out that these variants satisfy property \bone\ in part but nothing is known with respect to properties \btwo\ and \bthree.

\ignore{

journal version. refers to \cite{dcc/ChapinT98} 

In the same paper, the authors also propose a variant in which rotations are sorted following a scheme inspired by reflected Gray codes. This variant shows some improvements in terms of compression, but it has never been analyzed theoretically and it does not seem to support property~\bthree.}

To the best of our knowledge, the only non trivial \BWT variant that fully satisfies properties \bone--\bthree\ is the {Alternating \BWT} (\ABWT). This transformation has been derived in~\cite{GesselRestivoReutenauer2012} starting from a result in combinatorics of words~\cite{CDP2005} characterizing the \BWT as the inverse of a known bijection between words and multisets of primitive necklaces~\cite{GeRe}. The \ABWT is defined as the \BWT except that when sorting rotation instead of the standard lexicographic order we use a different lexicographic order, called the {\em alternating} lexicographic order.  In the alternating lexicographic order, the first character of each rotation is sorted according to the standard order of $\A$ (i.e., $a  < b < c$). However, if two rotations start with the same character we compare their second characters using the reverse ordering (i.e., $c<b<a$) and so on alternating the standard and reverse orderings in odd and even positions. Figure \ref{fig:ABWTx} (right) shows how the rotations of an input string are sorted using the alternating ordering and the resulting \ABWT.

\begin{figure}[tb]
{
{\small
$$
\begin{array}{llllllllll}
 & F & & & &          &  &      &    & L \\
 &\downarrow& & & &&&& &\downarrow\\
 & a & a & a & b & a & c & a & a & b \\ s\rightarrow
 & a & a & b & a & a & a & b & a & c \\ 
 & a & a & b & a & c & a & a & b & a \\ 
 & a & b & a & a & a & b & a & c & a \\ 
 & a & b & a & c & a & a & b & a & a \\
 & a & c & a & a & b & a & a & a & b \\ 
 & b & a & a & a & b & a & c & a & a \\ 
 & b & a & c & a & a & b & a & a & a \\ 
 & c & a & a & b & a & a & a & b & a  
\end{array}
\qquad\qquad
\begin{array}{llllllllll}
 & F & & & &          &  &      &    & L \\
 &\downarrow& & & &&&& &\downarrow\\
 & a & c & a & a & b & a & a & a & b \\ 
 & a & b & a & c & a & a & b & a & a \\
 & a & b & a & a & a & b & a & c & a \\ 
 & a & a & a & b & a & c & a & a & b \\ s\rightarrow
 & a & a & b & a & a & a & b & a & c \\ 
 & a & a & b & a & c & a & a & b & a \\ 
 & b & a & a & a & b & a & c & a & a \\ 
 & b & a & c & a & a & b & a & a & a \\ 
 & c & a & a & b & a & a & a & b & a  
\end{array}
$$
}
}
\caption{The original \BWT matrix for the string $s=aabaaabac$ (left), and the \ABWT matrix of cyclic rotations sorted using the alternating lexicographic order (right).
In both matrices the horizontal arrow marks the position of the original string $s$, and the last column $L$ is the output of the transformation.}\label{fig:ABWTx}
\end{figure}


The algorithmic properties of the \BWT and \ABWT are compared in~\cite{dlt/GiancarloMRRS18}. It is shown that they can be both computed and inverted in linear time and that their main difference is in the definition of the LF-map, i.e. the correspondence between the characters in the first and last column of the sorted rotations matrix. In the original \BWT the $i$-th occurrence of a character $c$ in the first column $F$ corresponds to the $i$-th occurrence of $c$ in the last column $L$. Instead, in the \ABWT the $i$-th occurrence of $c$ from the {\em top} in $F$ corresponds to the $i$-th occurrence of $c$ from the {\em bottom} in $L$. Since this modified LF-map can be still computed efficiently using \rank operations, the \ABWT can replace the \BWT for the construction of self-indices.

\ignore{The experiments in~\cite{dlt/GiancarloMRRS18} show that the \ABWT has essentially the same compression performance of the \BWT; we are not aware of any experiment using the \ABWT for indexing purposes. }

\ignore{


The algorithmic properties of the \BWT and \ABWT are compared in~\cite{dlt/GiancarloMRRS18}. It is shown that they can be both computed and inverted in linear time and that their main difference is in the definition of the LF-map, i.e. the correspondence between the characters in the first and last column of the sorted rotations matrix. In the original \BWT the $i$-th occurrence of a character $c$ in the first column $F$ corresponds to the $i$-th occurrence of $c$ in the last column $L$, i.e., equal characters appear in the same relative order in $F$ and $L$. Instead, in the \ABWT equal characters appear in the {\em reverse} order in $F$ and $L$, that is, the $i$-th occurrence of $c$ from the {\em top} in $F$ corresponds to the $i$-th occurrence of $c$ from the {\em bottom} in $L$. Since this modified LF-map can be still computed efficiently using \rank operations, the \ABWT can replace the \BWT for the construction of self-indices. The experiments in~\cite{dlt/GiancarloMRRS18} show that the \ABWT has essentially the same compression performance of the \BWT; we are not aware of any experiment using the \ABWT for indexing purposes. 

Note that in~\cite{dlt/GiancarloMRRS18} the \ABWT has been studied within a larger class of transformations in which the alphabet ordering depends on the position of the characters within any cyclic rotation. Although the ``compression boosting''  property holds for all transformations in this class, in~\cite{dlt/GiancarloMRRS18} is shown that the algorithmic techniques that allow us to invert the \ABWT in linear time cannot be applied to any other transformation in that class~\cite[Theorem~5]{dlt/GiancarloMRRS18}.

}

\section{BWTs based on Context Adaptive Alphabet Orderings}\label{sec:genBWT}

In this section we introduce a class of string transformations that generalize the \BWT in a very natural way. Given a primitive string $s$, as in the original \BWT definition, we consider the matrix containing all its cyclic rotations. 
In the original \BWT the matrix rows are sorted according to the standard lexicographic order. 
We generalize this concept by sorting the rows using an ordering that {\it depends on their common context}, i.e., their common prefix. Formally, for each string $x$ that prefixes two of more rows, we assume that an ordering $\pi_x$ is defined on the symbols of $\A$. When comparing two rows which are both prefixed by $x$, their relative rank is determined by the ordering $\pi_x$. Once the matrix rows have been ordered with this procedure, the output of the transformation is the last column of the matrix as in the original \BWT. 
Thus, these \BWT variants are based on {\em context adaptive alphabet orderings}. 
For simplicity in the following we call them {\em context adaptive \BWTs}.

An example is shown in Fig.~\ref{fig:BWTx}: the ordering associated to the empty string $\epsilon$ is $\pi_\epsilon = (b,a,c)$ so, among the rows that have no common prefix, first we have those starting with $b$, then those starting with $a$, and finally the one starting with $c$. Since $\pi_a = (c,a,b)$, among the rows which have $a$ as their common prefix, first we have the one starting with $c$, then the ones starting with $a$, followed by the ones starting with $b$. The complete ordering of the rows is established in a similar way on the basis of the orderings $\pi_x$. 

\begin{figure}[tb]
{
{\small
$$
\begin{array}{llllllllll}
 & F & & & &          &  &      &    & L \\
 &\downarrow& & & &&&& &\downarrow\\
 & b & a & a & a & b & a & c & a & a \\ 
 & b & a & c & a & a & b & a & a & a \\ 
 & a & c & a & a & b & a & a & a & b \\ s\rightarrow
 & a & a & b & a & a & a & b & a & c \\ 
 & a & a & b & a & c & a & a & b & a \\ 
 & a & a & a & b & a & c & a & a & b \\ 
 & a & b & a & a & a & b & a & c & a \\ 
 & a & b & a & c & a & a & b & a & a \\
 & c & a & a & b & a & a & a & b & a \\ 
\end{array}
$$
}
}
\caption{The generalized \BWT matrix for the string $s=aabaaabac$ computed using the orderings $\pi_\epsilon = (b,a,c)$, $\pi_a = (c,a,b)$, $\pi_{aa} = (c,b,a)$, and $\pi_x = (a,b,c)$ for every other substring $x$. The horizontal arrow marks the position of the original string $s$; the last column $L$ is the output of the transformation.}\label{fig:BWTx}
\end{figure}

We denote by $\Mx(s)$ the matrix obtained using this generalized sorting procedure, and by $L=\BWTx(s)$ the last column of $\Mx(s)$. Clearly $L$ depends on $s$ and the ordering used for each common prefix. Since we can arbitrarily choose an alphabet ordering for any substring $x$ of $s$, and there are $\asize!$ orderings to choose from, our definition includes a very large number of string transformations.
This class of transformations has been mentioned in~\cite{FGMS2005}[Sect.~5.2] under the name of {\it string permutations realized by a Suffix Tree} (the definition in~\protect{\cite{FGMS2005}} is slightly more general; for example it includes the bounded context \BWT, which is not included in our class). Indeed, if the input string $s$ has a unique end-of-string terminator, one can easily see that these transformations can be obtained assigning an ordering to the children of each node of the suffix tree of $s$

\ignore{(if $s$ doesn't have a unique terminator we simply replace the suffix tree with the compressed trie containing all cyclic rotations of $s$).}

\begin{figure}
\begin{center}
\includegraphics[scale=0.70]{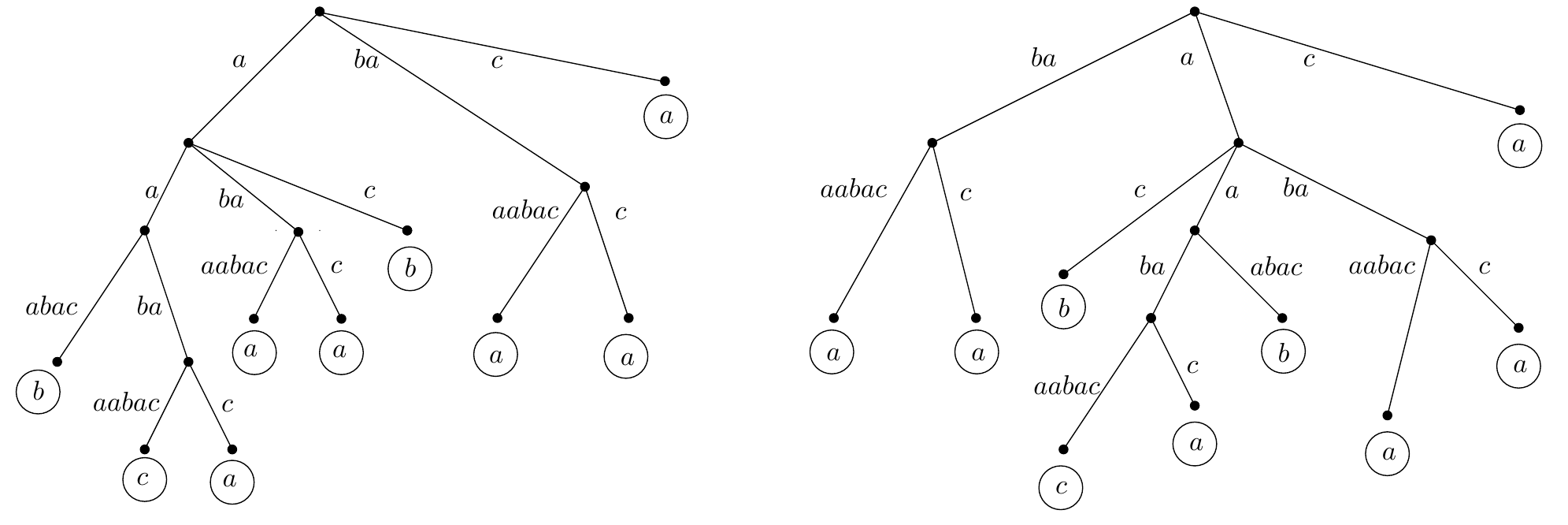}
\end{center}
\caption{Standard suffix tree for $s=aabaaabac$ with the symbol $c$ used as a string terminator (left), and suffix tree with edges reordered using the same orderings of Figure~\protect{\ref{fig:BWTx}} (right). To each leaf it is associated the symbol preceding in $s$ the suffix spelled by that leaf. Note that reading left to right the symbols associated to each leaf gives $\BWT(s)$ (left) and $\BWTx(s)$ (right).\label{fig:suffixTree}}
\end{figure}

Although in~\cite{FGMS2005} the authors could not prove the invertibility of context adaptive transformations, which we do in Section~\ref{subsec:inversion}, they observed that their relationship with the suffix tree has two important consequences: 1) they can be computed in $\Oh(n\log\sigma)$ time with a proper suffix tree visit (see Fig.~\ref{fig:suffixTree}), and 2) they provably produce {\em highly compressible} strings, i.e., they have the ``boosting'' property of transforming a zeroth order compressor into a $k$-th order compressor.


To see that the generalized \BWTs can be computed in $\Oh(n\log\sigma)$ time consider first the simpler case in which the string $s$ has a unique end-of-string terminator. To build $L=\BWTx(s)$ we first build the suffix tree for $s$. Then, we visit the suffix tree in depth first order except that when we reach a node $u$ (including the root), we sort its outgoing edges according to their first characters using the permutation associated to the string $u_x$ labeling the path from the root to $u$. During such visit, each time we reach a leaf  we write the symbol associated to it: the resulting string is exactly $L=\BWTx(s)$. The above argument also shows that the number of permutations required to define a generalized \BWT is at most $|s|$, i.e. the number of internal suffix tree nodes. If $s$ doesn't have a unique terminator, the argument is analogous except that we replace the suffix tree with the compressed trie containing all the cyclic rotations of $s$. To see that generalized \BWTs have the boosting property we observe that the proof for the \BWT (Theorem~5.4 in~\cite{FGMS2005}) is based on structural properties of the suffix tree, and can be repeated verbatim for the generalized \BWTs.

Summing up, context adaptive transformations generalize the \BWT in two important aspects: efficient (linear time in $n$) computation and compressibility. In~\cite{FGMS2005} the only known instances of {\em reversible} suffix tree induced transformations were the original \BWT and the bounded context \BWT. In the following, we prove that {\em all} context adaptive \BWTs defined above are invertible. Interestingly, to prove invertibility we first establish another important property of these transformations, namely that they can be used to count the number of occurrences of a pattern in $s$, which is another fundamental property of the original \BWT.

We conclude this section observing that both the \BWT and \ABWT belong to the class we have just defined. To get the \BWT we trivially define $\pi_x$ to be the standard $\A$ ordering for every $x$, to get the \ABWT we define $\pi_x$ to be the standard $\A$ ordering for every $x$ with $|x|$ even, and the reverse ordering for $\A$ for every $x$ with $|x|$ odd. Indeed in the full paper we will show that the complete class of transformations studied in~\cite{dlt/GiancarloMRRS18} is a subclass of context adaptive transformations.

\subsection{Counting occurrences of patterns in Context Adaptive BWTs}\label{subsec:counting}

Let $L=\BWTx(s)$ denote a context adaptive \BWT. In the following we assume that $L$ is enriched with data structures supporting constant time rank queries as in Theorem~\ref{theo:rsa}. In this section we show that given $L$ and the set of alphabet permutations used to build $\Mx(s)$ then, for each string $x$, we can determine in $\Oh(\asize|x|^2)$ time the set of $\Mx(s)$ rows prefixed by $x$. We preliminary observe that by construction this set of rows, if non-empty, form a contiguous range inside $\Mx(s)$. This observation justifies the following definitions. 

\begin{definition}\label{def:range}
Given a string $x$, we denote by $\Rng{x} = [b_x,\ell_x]$ the range of rows of $\Mx(s)$ prefixed by $x$. More precisely, if $\Rng{x} = [b_x, \ell_x]$, then row $i$ is prefixed by $x$ if and only if it is $b_x \leq i < b_x + \ell_x$.  If no rows are prefixed $x$ we set $\Rng{x} = [0,0]$. Note that $\ell_x$ is the number of occurrences of $x$ in the circular string $s$.
\end{definition}

For technical reasons, given $x$, we are also interested in the set of rows prefixed by the strings $xc$ as $c$ varies in $\A$. Clearly, these sets of rows are consecutive in $\Mx(s)$ and their union coincides with $\Rng{x}$. 

\begin{definition}\label{def:rangex}
Given a string $x$, we denote by $\Rngx{x}$ the set of $\asize+1$ integers $[b_x,\ell_1, \ell_2, \ldots, \ell_\asize]$ such that $b_x$ is the lower extreme of $\Rng{x}$ and, for $i=1,\ldots,\asize$, $\ell_i$ is the number of rows of $\Mx(s)$ prefixed by $xc_i$.
\end{definition}

Since $\Rng{x}$ is the union of the ranges $\Rng{x\xchar}$ for $\xchar\in\A$, we have that if $\Rngx{x} = [b_x,\ell_1, \ell_2, \ldots, \ell_\asize]$, then $\Rng{x}=[b_x,\sum_i \ell_i]$. Note also that the ordering of the ranges $\Rng{x\xchar}$ within $\Rng{x}$ is determined by the permutation $\pi_x$. As observed in Section~\ref{sec:prel}, we can assume that $L$ supports constant time rank queries. This implies that in constant time we are also able to count the number of occurrences of a symbol $c$ inside a substring $L[i,j]$.

\begin{lemma}\label{lemma:rangex+sigma}
Given $\Rngx{x}$ and the permutation $\pi_x$, the set of values $\Rng{x\xchar_i}$ for all $\xchar_i\in\A$ can be computed  in $\Oh(\asize)$ time.
\end{lemma}

\begin{proof}
If $\Rngx{x} = [b_x,\ell_1, \ell_2, \ldots, \ell_\asize]$ then
$\Rng{x\xchar_i} = [b,\ell]$ with
\begin{equation}\label{eq:j<i}
b = b_x + \sum_j \ell_j, \qquad
\ell = \ell_i
\end{equation}
where the summation in~\eqref{eq:j<i} is done over all $j\in\{1,2,\ldots,\asize\}$ such that $\xchar_j$ is smaller than $\xchar_i$ according to the permutation $\pi_x$.
\end{proof}

\begin{lemma}\label{lemma:newton}
Let $x=x_1 x_2 \cdots x_m$ be any length-$m$ string with $m>1$. Then, given $\Rngx{x_1 \cdots x_{m-1}}$ and $\Rngx{x_2 \cdots x_m}$, the set of values $\Rngx{x_1 \cdots x_m}$ can be computed in $\Oh(\asize)$ time.
\end{lemma}

\begin{proof}
By Lemma~\ref{lemma:rangex+sigma}, given $\Rngx{x_1 \cdots x_{m-1}}$ and $x_m$, we can compute $\Rng{x_1 \cdots x_m} = [b_x,\ell_x]$. In order to compute $\Rngx{x_1 \cdots x_m}$, we additionally need the number of rows prefixed by $x_1 x_2 \cdots x_m \xchar$, for any $\xchar\in\A$. These numbers can be obtained by first computing the ranges $\Rng{x_2 \cdots x_m \xchar}$ using again Lemma~\ref{lemma:rangex+sigma}, and then counting the number of rows prefixed by $x_1 x_2 \cdots x_m \xchar$, counting the number of $x_1$ in the portions of $L$ corresponding to each range $\Rng{x_2 \cdots x_m \xchar}$. The counting takes $\Oh(\asize)$ time since we are assuming $L$ supports constant time \rank as in Theorem~\ref{theo:rsa}. 
\end{proof}

\begin{theorem}\label{theorem:rangex}
Suppose we are given $\BWTx(s)$ with constant time rank support, and the set of permutations used to compute the matrix $\Mx(s)$. Then, given any string $x=x_1 x_2 \cdots x_p$, the range of rows $\Rng{x}$ prefixed by $x$ can be computed in $\Oh(\asize p^2)$ time and $\Oh(\asize p)$ space. \ignore{Note that if $\Rng{x}=[b_x,\ell_x]$, then $\ell_x$ is the number of occurrences of $x$ in the circular string $s$.} 
\end{theorem}

\begin{proof}
We need to compute $\Rng{x_1 x_2 \cdots x_p}$. To this end we consider the following scheme, inspired by the Newton finite difference formula:
$$
\begin{array}{clllll}
\Rngx{x_1} & \Rngx{x_1 x_2} & \Rngx{x_1 x_2 x_3} & \cdots & \Rngx{x_1 x_2 \cdots x_{p-1}} & \Rngx{x_1 x_2 \cdots x_p} \\
\Rngx{x_2} & \Rngx{x_2 x_3} & \Rngx{x_2 x_3 x_4} & \cdots &  \Rngx{x_2 \cdots x_{p}} & \\
\Rngx{x_3} & \Rngx{x_3 x_4} &\multicolumn{1}{c}{ \cdots} & \\
\vdots \\
\Rngx{x_p}
\end{array}
$$
Using Lemma~\ref{lemma:newton} we can compute $\Rngx{x_i \cdots x_j}$ given $\Rngx{x_i \cdots x_{j-1}}$ and $\Rngx{x_{i+1} \cdots x_j}$. Thus, from two consecutive entries in the same column we can compute one entry in the following column.  
To compute $\Rng{x_1 x_2 \cdots x_p}$ we can for example perform the computation 
bottom-up, proceeding row by row. In this case we are essentially computing the ranges corresponding to $x_p$, $x_{p-1} x_p$, $x_{p-2} x_{p-1} x_p$ and so on, in a sort of backward search. However, we can also perform the computation top down, diagonal by diagonal, and in this case we are computing the ranges corresponding to $x_1$, $x_1 x_2$, and so on up to $x_1 \cdots x_p$. In both cases, the information one need to store from one iteration to the next is $\Oh(p)$ $\Rngx{\cdot}$ values, which take $\Oh(\asize p)$ words. By Lemma~\ref{lemma:newton}, the computation of each value takes $\Oh(\asize)$ time so the overall complexity is $\Oh(\asize p^2)$ time.  
\end{proof}


Note that our scheme for the computation of $\Rng{x}$ is based on the computation of $\Rngx{y}$ for $\Oh(p^2)$ substrings $y$ of $x$. If $x$ has many repetitions, the overall cost could be less than quadratic. In the extreme case, $x=a^p$, $\Rng{x}$ can be computed in $\Oh(\asize p)$ time.

\subsection{Inverting Context Adaptive BWTs}\label{subsec:inversion}

We now show that the machinery we set up for counting occurrences can be used to retrieve $s$ given $\BWTx(s)$, thus to invert any context adaptive \BWT. 

\begin{lemma}\label{lemma:nextchar}
Given $\Rngx{x} = [b_x,\ell_1,\ell_2,\ldots,\ell_\asize]$ and a row index $i$ with $b_x \leq i < b_x + \sum_{j=1}^\asize \ell_j$,  the $(|x|+1)$-st character of row $i$ can be computed in $\Oh(\asize)$ time.
\end{lemma}

\begin{proof}
Let $\rho_{1}, \ldots, \rho_{\asize}$ denote the alphabet symbol reordered according to the permutation~$\pi_x$, and let $\ell'_1, \ldots, \ell'_\asize$ denote the values $\ell_1,\ldots,\ell_\asize$ reordered according to the same permutation. Since $i\in\Rng{x}$, row $i$ is prefixed by $x$. Since the rows prefixed by $x$ are sorted in their $(|x|+1)$-st position according to $\pi_x$, the $(|x|+1)$-st symbol of row $j$ is the symbol $\rho_j$ such that 
$$ 
b_x + \sum_{1 \leq h<j} \ell_{h} \;\leq\; i \;<\;  
b_x + \sum_{1 \leq h\leq j} \ell_{h}
$$
\end{proof}

\begin{theorem}\label{theo:inversion}
Given $\BWTx(s)$ with constant time rank support, the permutations $\pi_x$ used to build the matrix $\Mx(s)$, and the row index $i$ containing $s$ in $\Mx(s)$, the original string $s$ can be recovered in $\Oh(\asize |s|^2)$ time and $\Oh(\asize |s|)$ working space.
\end{theorem}

\begin{proof}
Let $s = s_1 s_2 \cdots s_n$. From $\BWTx(s)$, in $\Oh(n)$ time we retrieve the number of occurrences of each character in $s$ and hence the ranges $\Rng{\xchar_1}$, $\Rng{\xchar_2}$, \ldots, $\Rng{\xchar_\asize}$. From those and the row index $i$, we retrieve $s$'s first character $s_1$. Next, counting the number of occurrences of $s_1$ in the ranges of $\BWTx(s)$ corresponding to $\Rng{\xchar_1}$, $\Rng{\xchar_2}$, \ldots, $\Rng{\xchar_\asize}$, we compute $\Rngx{s_1}$.

Finally, we show by induction that, for $m=1,\ldots,n-1$, given $\Rngx{s_1 s_2 \cdots s_m}$, we can retrieve $s_{m+1}$ and $\Rngx{s_1 s_2 \cdots s_{m+1}}$ in $\Oh(m\asize)$ time. By Lemma~\ref{lemma:nextchar}, from $\Rngx{s_1 s_2 \cdots s_m}$ and $i$ we retrieve $s_{m+1}$. Next, assuming we maintained the ranges 
$\Rngx{s_j \cdots s_m}$, for $j=1,\ldots,m$ we can compute $\Rngx{s_j \cdots s_{m+1}}$ adding one diagonal to the scheme shown in the proof of Theorem~\ref{theorem:rangex}. By Lemma~\ref{lemma:newton}, the overall cost is $\Oh(\asize |s|^2)$ as claimed.
\end{proof}

\ignore{The above theorem establishes that all context adaptive \BWTs are invertible. Note that in our definition, the alphabet ordering $\pi_x$ associated to $x$ can depend on the whole string $x$; in this sense the context has full memory. In the next two sections we study two subclasses in context has a bounded memory: in Section~\ref{sec:abwt} the ordering $\pi_x$ only depends on $|x| \bmod k$, in Section~\ref{sec:local} $\pi_x$ only depends on the the last $k$ characters of $x$.}

\ignore{We consider this an important conceptual result. However, from a practical point of view a transformation whose definition requires the specification of $\Oh(|s|)$ alphabet permutations appears rather cumbersome. For this reason, in the rest of the paper, we consider some subclasses of transformations in which the permutations associated to each prefix are defined by ``simple'' rules. We show that in some cases non trivial generalized \BWTs have simpler and more efficient inversion algorithms.}

\ignore{---- versione journal

\subsection{Special case: reversal ordering}

Given a permutation $\pi$ of the alphabet $\A$ we denote by $\pi^R$ the reversal of $\pi$, that is, the permutation such that for each pair of symbols $\xchar_i$, $\xchar_j$ in $\A$
$$
\pi^R(\xchar_i) < \pi^R(\xchar_j) \quad\Longleftrightarrow\quad 
\pi(\xchar_i) > \pi(\xchar_j).
$$


Let $\pi$ by an arbitrary permutation of $\A$. We consider the subclass of transformations in which the permutation $\pi_x$ associated to each substring $x$ can be either $\pi$ or its reversal $\pi^R$. Once $\pi$ is established, for each string $x$ we only need an additional bit to specify the ordering $\pi_x$. For this subclass we say that the matrix $\Mx(s)$ is based on a {\em $\pm$ ordering}. In this section we show that for the transformations in his subclass the space and time for inversion can be reduced by a factor $\asize$.

\begin{lemma}\label{lemma:newton+-}
Let $\Mx(w)$ be based on a $\pm$ ordering, and let
$x=x_1 x_2 \cdots x_m$ be any length-$m$ string with $m>1$. 
Then, given $\Rng{x_1 \cdots x_{m-1}}$, $\Rng{x_2 \cdots x_m}$ and $\Rng{x_2\cdots x_{m-1}}$ we can compute $\Rng{x_1 \cdots x_m}$ in $\Oh(1)$ time.
\end{lemma}

\begin{proof}
Let $y=x_1 \cdots x_{m-1}$ and $z=x_2 \cdots x_{m-1}$. Recall that there is a bijection between the rows in $\Rng{y}$ and the rows in $\Rng{z}$ whose last symbol is $x_1$. We exploit this bijection to find $\Rng{y x_m}$ given $\Rng{z x_m}$. The size of $\Rng{y x_m}$ is equal to the number of rows in $\Rng{z x_m}$ ending with $x_1$, so to completely determine $\Rng{y x_m}$ we just need to compute how many rows in $\Rng{y}$ precedes $\Rng{y x_m}$. If $\pi_y = \pi_z$ this number is equal to the number of rows in $\Rng{z}$ {\em above} $\Rng{z x_m}$ and ending with $x_1$. If  $\pi_y \neq \pi_z$, since necessarily $\pi_y = \pi_z^R$, this number is equal to the number of rows in $\Rng{z}$ {\em below} $\Rng{z x_m}$ and ending with $x_1$. Since counting the number of rows ending in $x_1$ in a given range can be done using rank operations on $\BWTx(s)$ the lemma follows. 
\end{proof}

\begin{lemma}\label{lemma:rangex+-}
Suppose $\Mx(s)$ is based on a $\pm$ ordering. Given $\BWTx(s)$ with constant time rank support, and any string $x=x_1 x_2 \cdots x_p$ we can compute, in $\Oh(p^2)$ time and $\Oh(p)$ space, the range of rows prefixed by $x$.  
\end{lemma}

\begin{proof}
We reason as in the proof of Theorem~\ref{theorem:rangex} except that because of Lemma~\ref{lemma:newton+-} we work with $\Rng{\cdot}$ instead of $\Rngx{\cdot}$. The thesis follows observing that each entry takes $\Oh(1)$ space and can be computed in $\Oh(1)$ time.
\end{proof}

\begin{theorem}\label{theo:inversion+-}
Suppose $\Mx(s)$ is based on a $\pm$ ordering. Given $\BWTx(s)$ with constant time rank support and the row index $i$ containing $s$ in $\Mx(s)$, we can retrieve the original string $s$ in $\Oh(|s|^2)$ time and $\Oh(|s|)$ working space.\qed
\end{theorem}
}

\section{BWTs based on local orderings}\label{sec:local}
In our definition of context adaptive transformation, the alphabet 
ordering $\pi_x$ associated to $x$ can depend on the whole string $x$; in this sense the context has full memory. In this section we consider transformations in which the context has a {bounded memory}, in that it only depends on the last $k$ symbols of $x$, where $k$ is fixed. In the following we refer to these string transformations as {\em \BWTs based on local orderings}.  

We start by analyzing the case $k=1$. For such local ordering transformations the matrix $\Mx(s)$ depends on only $\asize+1$ alphabet orderings: one for each symbol plus the one used to sort the first column of $\Mx(s)$. The following lemma establishes an important property of local ordering transformations. 

\ignore{

In this section we show that local ordering transformations have properties very similar to the original \BWT.  

Note that the only intersection between this class and the class of local ordering transformations, where rows' relative order depends only on the last character of their longest common prefix, is the transformation in which the same alphabet ordering is associated to every string $x$, which is essentially a standard \BWT with a reordered alphabet
}

\begin{lemma}\label{lemma:almostbwt}
If $\Mx(s)$ is based on a local ordering, then for any pair of characters $x_1$, $x_2$  there is an order preserving bijection between the set of rows starting with $x_1x_2$ and the set of rows starting with $x_2$ and ending with $x_1$.
\end{lemma}

\begin{proof}
Note that both sets of rows contain a number of elements equal to the number of occurrences of $x_1 x_2$ in the circular string $s$. In the following, we write $s[i\cdots]$ to denote the cyclic rotation of $s$ starting with $s[i]$. Assume that rotations $s[i\cdots]$ and $s[j\cdots]$ both start with $x_2$ and end with $x_1$ and let $h$ denote the first column in which the two rotations differ. Rotation $s[i\cdots]$ precedes $s[j\cdots]$ in $\Mx(s)$ if and only if $s[i+h]$ is smaller than $s[j+h]$ according to the alphabet ordering associated to symbol $s[i+h-1]=s[j+h-1]$. The two rotations $s[i-1 \cdots]$ and $s[j-1 \cdots]$ both start with $x_1 x_2$ and their relative position also depends on the relative ranks of $s[i+h]$ and $s[j+h]$ according to the alphabet ordering associated to symbol $s[i+h-1]=s[j+h-1]$. Hence the relative order of  $s[i-1 \cdots]$ and $s[j-1 \cdots]$ is the same as the one of $s[i\cdots]$ and $s[j\cdots]$.
\end{proof}

Armed with the above lemma, we now show that for local ordering transformations we can establish much stronger results than the one provided in Section~\ref{subsec:counting}.

\begin{lemma}\label{lemma:newton:local}
Suppose $\BWTx(s)$ supports constant time rank queries. Let $x=x_1 x_2 \cdots x_m$ be any length-$m$ string with $m>1$. Then, given $\Rng{x_1 x_2}$, $\Rng{x_2}$ and $\Rng{x_2 \cdots x_m}$, the value $\Rng{x_1 \cdots x_m}$ can be computed in $\Oh(1)$ time.
\end{lemma}

\begin{proof}
By Lemma~\ref{lemma:almostbwt} there is an order preserving bijection between the rows in $\Rng{x_1 x_2}$ and those in $\Rng{x_2}$ ending with $x_1$. In this bijection, the rows in $\Rng{x_1 \cdots x_m}$ correspond to those in $\Rng{x_2 \cdots x_m}$ ending with $x_1$. Hence, if, among the rows starting with $x_2$ and ending with $x_1$, those prefixed by $x_2\cdots x_m$ are in positions $r, r+1, \ldots, r+h$, then, among the rows starting with $x_1 x_2$, those prefixed by $x_1 x_2 \cdots x_m$ are in positions $r,r+1,\ldots,r+h$.
\end{proof}

\begin{theorem}\label{theorem:rangex:local}
Suppose $\BWTx(s)$ is based on a local ordering and supports constant time rank queries. After a $\Oh(\asize^2)$ time preprocessing, given any string $x=x_1 x_2 \cdots x_p$, the range of rows prefixed by $x$ can be computed in $\Oh(p)$ time and $\Oh(p)$ space. 
\end{theorem}

\begin{proof}
We reason as in the proof of Theorem~\ref{theorem:rangex}, except that because of Lemma~\ref{lemma:newton:local} we can work with $\Rng{\cdot}$ instead of $\Rngx{\cdot}$ and we only need to compute the first two columns and the diagonal. In the preprocessing step, we compute $\Rng{\xchar_i}$ and $\Rng{\xchar_i\xchar_j}$ for any pair $(\xchar_i, \xchar_j) \in \A^2$. During the search phase, we compute each diagonal entry in constant time. 
\end{proof}

Another immediate consequence of Lemma~\ref{lemma:almostbwt} is that we can efficiently ``move back in the text'' as in the original \BWT. Note this operation is the base for \BWT inversion and for snippet extraction and locate operations on FM-indices~\cite{Ferragina:2005}.

\begin{lemma}
Suppose $\BWTx(s)$ is based on a local ordering and supports constant time rank and access queries. Then, after a $\Oh(\asize^2)$ time preprocessing, given a row index $i$ we can compute in $\Oh(1)$ time the index of the row obtained from the $i$-th row with a circular right shift by one position.
\end{lemma}

\begin{proof}
Compute the first and last symbol of row $i$ and then apply Lemma~\ref{lemma:almostbwt}.
\end{proof}

\begin{corollary}
If $\BWTx(s)$ is based on a local ordering and supports constant time rank and access queries, $\BWTx(s)$ can be inverted in $\Oh(\asize^2+|s|)$ time and $\Oh(\asize^2)$ working space.\qed
\end{corollary}

In the full paper we will show that bounded context adaptive \BWTs can be generalized to the case in which the ordering $\pi_x$ depends only on the last $k>1$ symbols of $x$. Search and inversion can still be performed in linear time with the only difference that preprocessing now takes $\Oh(\asize^{k+1})$ time and space. 

\ignore{We conclude this section showing an alternative way to derive bounded context adaptive transformations. Considering for simplicity the case $k=1$, and assume the transformation $\BWTx$ is defined by the $\sigma+1$ orderings $\pi_\epsilon$ and $\pi_c$ for $c\in\A$.  Consider now the ordering $\Pi$ over $\A^2 = \A \times \A$ defined as follows: Given the pairs $x_1x_2$, $y_1y_2$ in $\A^2$ it is $(x_1x_2 <_\Pi y_1y_2)$ iff $(x_1\neq y_1,\; x_1 <_{\pi_\epsilon} y_1)$, or $(x_1\!=y_1\!=c,\; x_2 <_{\pi_c} y_2)$. We now show that $\BWTx$ is equivalent to the original \BWT over the alphabet $\A^2$ ordered according to~$\Pi$.

To each string $s$ we associate a new string $S$ over $\A^2$,  defined by $S[i] = s[i]s[i+1]$ with indices taken modulo $|s|$. For example, for $s=aabaaabc$ it is $S= aa\, ab\, ba\, aa\, aa\, ab\, bc\, ca$. There is a natural correspondence between rotations of $s$ and $S$, and because of the definition of $\Pi$ the ordering of $s$'s rotations in $\Mx(s)$ coincides with the ordering of the corresponding rotations of $S$ in $M(S)$. As a consequence, if $\BWT(S)$ (the last column of $M(S)$) has the form $\BWT(S) = x_1y_1\,x_2y_2\cdots x_ny_n$, we have that $\BWTx(s) = x_1 x_2 \cdots x_n$, and the first column of $\Mx(s)$ is $y_1 y_2 \cdots y_n$. Note that the traditional order preserving LF-map for $M(S)$ provides us an alternative proof of Lemma~\ref{lemma:almostbwt}.

The above alternative view has probably no practical interest: there is no need to work with the alphabet $\A^2$ to emulate something we can do working over $\A$. However, from the theoretical point of view it is certainly intriguing, and deserving further investigation, the fact a family of \BWT variants can be obtained by first transforming the string and the alphabet and then applying the standard \BWT followed by the string back-transformation.}

\newcommand{\Run}{\rho}
\newcommand{\runmin}{Opt}
\newcommand{\bw}{bw}

\section{Run minimization problem}\label{sec:runs}

In this section we consider the following problem: given a string $s$ and a class of \BWT variants, find the variant that minimizes the number of runs in the transformed string. As we mentioned in the introduction this problem is relevant for the compression of highly repetitive collections. 


We consider the general class of context adaptive \BWTs described in Section~\ref{sec:genBWT}. In this class we can select an alphabet ordering $\pi_x$ independently for every substring~$x$. However, it is easy to see that the only orderings that influence the output of the transform are those  associated to strings corresponding to the internal nodes of the suffix tree of~$s$.  Given a suffix tree node $v$ we denote by $\bw(v)$ the multiset of symbols associated to the leaves in the subtree rooted at~$v$. We say that a string $z_v$ is a {\em feasible} arrangement of $\bw(v)$ if we can reorder the nodes in the subtree rooted at $v$ so that $z_v$ is obtained reading left to right the symbols in the reordered subtree. 
For example, in the suffix tree of Fig.~\ref{fig:suffixTree} (left), if $v$ is the internal node with upward path $aa$ it is $\bw(v) = \{a,b,c\}$ and both $bac$ and $cab$ are feasible arrangements of $\bw(v)$, while $abc$ is {\em not} a feasible arrangement.  If $\tau$ is the suffix tree root, using the above notation our problem becomes that of finding the feasible arrangement of $\bw(\tau)$ with the minimal number of runs. For constant alphabets the following theorem shows that the optimal arrangement can be found in linear time using dynamic programming. 

\ignore{Hence, our problem is to assign an ordering to each internal node as to minimize the number of runs in the output.}

\ignore{Given a suffix tree node $v$ we write $\bw(v)$ to denote the string obtained reading left to right the symbols associated to the leaves descendant from $v$. Hence if $\tau$ is the suffix tree root $\bw(\tau) = \BWT(s)$, see Figure~\ref{fig:suffixTree}. We say that permutation of $\bw(v)$ is {\em feasible} if we can reorder to the nodes in the subtree rooted at $v$ such that reading left to right the symbols in the reordered subtree}

\ignore{With a little abuse of notation, we write $\pi(v)$ to denote the ordering associated to node $v$, and the problem we consider becomes to assign an ordering to each internal node as to minimize the number of runs in the output. }

\begin{theorem}\label{theo:runs}
Given a string $s$ over a constant size alphabet, the context adaptive transformation $\BWTx$ minimizing the number of runs in $\BWTx(s)$ can be found in $\Oh(|s|)$ time.
\end{theorem}

\begin{proof}
Let $\runmin$ denote the minimal number of runs. We show how to compute $\runmin$ with a dynamic programming algorithm; the computation of the alphabet orderings giving $\runmin$ is done using standard techniques. For each suffix tree node $v$ and pairs of symbols $c_i$, $c_j$ let $\Run(v,c_i,c_j)$ denote the minimal number of runs among all feasible arrangements of $\bw(v)$ starting with $c_i$ and ending with $c_j$. Clearly, if $\tau$ is the suffix tree root, then $\runmin = \min_{i,j}\Run(\tau,c_i,c_j)$.

For each leaf $\ell$ it is $\Run(\ell,c_i,c_j)=1$ if $c_i=c_j=\bw(\ell)$ and  $\Run(\ell,c_i,c_j)=\infty$ otherwise. We need to show how to compute, for each internal node $v$, the $\sigma^2$ values $\Run(v,c_i,c_j)$ for $c_i$, $c_j$ in $\A$, given the, up to $\sigma^3$ values, $\Run(w_k,c_\ell,c_m)$, $k=1,\ldots,h$, where  $w_1, \ldots, w_h$ are the children of $v$. To this end, we show that for each ordering $\pi$ of $w_1, \ldots,w_h$ we can compute in constant time the minimal number of runs among all the feasible arrangements of $\bw(v)$ starting with $c_i$ and ending with $c_j$ and with the additional constraint that $v$'s children are ordered according to $\pi$. 

To simplify the notation assume $w_1,\ldots, w_h$ have been already reordered according to $\pi$. For $k=1,\ldots,h$ let $M_\pi[k,c_\ell,c_m]$ denote the minimal number of runs among all strings $x$ such that $x = y_1 \cdots y_k$ where $y_t$, for $t=1,\ldots,k$, is a feasible arrangement of $\bw(w_t)$, and with the additional constraints that $y_1$ starts with $c_\ell$ and $y_k$ ends with $c_m$. We have
$$
M_\pi[1,c_\ell,c_m] = \Run(w_1,c_\ell,c_m)
$$
and for $k=2,\ldots,h$
\begin{equation}\label{eq:Mpi}
M_\pi[k,c_\ell,c_m] = \min_{i,j} \left(M_\pi[k-1,c_\ell,c_i]+\Run(w_k,c_j,c_m)-\delta_{ij}\right)
\end{equation}
where $\delta_{ij}=1$ if $i=j$ and $0$ otherwise. Essentially, \eqref{eq:Mpi} states that to find the minimal number of runs for $w_1, \ldots,w_k$ we consider all possible ways to combine an optimal solution for $w_1, \ldots,w_{k-1}$ followed by a feasible arrangement of $\bw({w_k})$. The $\delta_{ij}$ term comes from the fact that the number of runs in the concatenation of two strings is equal to the sum of the runs in each string, minus one if the last symbol of the first string is equal to the first symbol of the second string. 

Once we have the values $M_\pi[h,c_i,c_j]$, the desired values $\Run(v,c_i,c_j)$ are obtained taking the minimum over all possible alphabet ordering~$\pi$.
\end{proof}

Clearly the above theorem does not immediately yield a practical compressor, since the cost of specifying the alphabet ordering at each node is likely to outweigh the advantage of minimizing the number of runs. However we notice that: 1) the optimal transformation for a string will reasonably produce good results on similar strings so we can compute and store the ordering once and use it many times, 2) since Theorem~\ref{theo:runs} holds for the most general class, it provides a lower bound for the more interesting and practical \BWTs based on local orderings and the \ABWT.

\ignore{; such lower bound can play an important role both in practical and theoretical analysis.  

We conclude this section observing that, for the classes of \BWTs based on local orderings with constant parameter $k$, the number of possible orderings only depends on the alphabet size. Hence, for constant alphabet determining the variant that minimizes the number of runs can be done in linear time by testing all the alternatives.  However, since the number of possible orderings grows exponentially with the alphabet size, such enumerative approach is unpractical: we leave the development of polynomial algorithms for non constant alphabets as an open problem for future research. }

\section{Conclusions and Future Directions of Research}

In this paper we introduced a new class of string transformations and shown that they have the same remarkable properties of the \BWT: they can be computed and inverted in linear time, they have the ``compression boosting'' property, and they support linear time pattern search directly in the compressed text. This implies that such transformations can replace the \BWT in the design of self-indices without any asymptotic loss of performance. Given the crucial role played by the \BWT even outside the area of string algorithms, we believe that expanding the number of efficient \BWT variants can lead to theoretical and practical advancements. A natural consequence will be the design of ``personalized'' transformations, where one will choose the ``best'' alternative to the \BWT according to costs and benefits dictated by application domains. As an example, motivated by the problem of compressing highly repetitive string collections that arises in areas such as Bioinformatics, we considered the problem of determining the \BWT variant that minimizes the number of runs in the transformed string. 

Our efficient \BWT variants are a special case of a more general class of transformations that have the same properties of the \BWT but for which we could not devise efficient (linear time) inversion and search algorithms. This more general class includes also the \ABWT the only known transformation which has the same properties of the original \BWT. We believe this larger class of transformation should be further investigated. For example, it would be interesting to analyze the maximum compression achievable using transformations in that class, since that would give an upper bound also for the more practical variants. Also, it would be worthwhile to investigate if there are subclasses for which the inversion and search operations could be done in less than quadratic time. As we observed immediately after the proof of Theorem~\ref{theorem:rangex} in Section~\ref{subsec:counting}, it appears to be possible to reduce the cost of the search operation taking advantage of the structure of the searched pattern.


\ignore{The fact that, at least in one important domain,  an optimal choice can be determined by an algorithm  taking linear time  is also very significant, given the richness of hardness results for Data Compression formulated as an optimization problems even when the alphabet has size two or three, e.g.,\cite[Section A4.2]{gandj:gj}.
Finally, experimental testing of the new ideas that have emerged here is planned.}

\ignore{
The main result of this paper is the introduction of a new family of \BWT variants, based on the concept of local orderings, that share all the \vir{magic and useful} properties of the original \BWT: they can be computed and inverted in linear time, they produce provably highly compressible strings, and they support pattern search directly in the transformed string in optimal linear time. The paper also shows that this family belongs to a more general class of invertible transformations guaranteed to produce highly compressible strings. This more general class also includes the Alternating \BWT, the only other known \BWT variant with all the properties of the original \BWT. It goes without saying that this general class introduced here could include other interesting subclasses of \BWT variants: this is the first open problem that we leave for further study. In the full paper we will describe another subclass of transformations that could be competitive for strings over large alphabets. 

Another area of research originated by the above results is the design of \vir{personalized} transformations. Indeed, having at our disposal a large number of alternatives to the \BWT suggests that it could be worthwhile to choose the transformation more suitable for a given input or for a given data structure. In this paper, motivated by the problem of compressing highly repetitive collections, we initiated this study by considering the problem of determining the \BWT variant that minimizes the number of runs in the transformed string. Given the large number of data structures featuring the \BWT we believe that this area of research can have significant impact in theory and practice. }

\ignore{It goes without saying that having found two subclasses satisfying properties \bone--\bthree\ suggests that there can be other interesting subclasses: this is the main open problem we leave for further study.}


\end{document}